\documentclass[12pt]{article}
\usepackage{amsfonts}
\usepackage{upgreek}
\usepackage{amsmath}
\usepackage{amssymb, amscd, amsthm}
\usepackage[all]{xy}
\usepackage[dvips]{graphicx}
\usepackage{verbatim}
\newtheorem{theo}[]{{\emph{Theorem}}}

\theoremstyle{remark}

\theoremstyle{definition}

\usepackage{amsmath}
\usepackage{amssymb, amscd, amsthm,amsfonts}
\usepackage[all]{xy}
\usepackage[dvips]{graphicx}
\usepackage{verbatim}
\usepackage[perpage,symbol*]{footmisc}

\newtheorem{lem}{Lemma}

\newcommand{\Tr}{\mathrm{Tr}}

\newcommand{\GF}{\mathrm{GF}}

\newcommand{\de}{\delta}

\newcommand{\ra}{\rightarrow}

\begin{document}

\title{Binary sequences with three-valued cross correlations of different lengths}
 \maketitle
\begin{center}
\author{
 \text{Jinquan
\;\;Luo$^*$ } } \footnotetext{The author is with School of
Mathematics and Statistics, Central China Normal University, Wuhan,
Hubei, China 430079. This work is supported by NSFC under grant
11471008.
\par  E-mail addresses: luojinquan@mail.ccnu.edu.cn}
\end{center}

{\bf\it Abstract--} \textbf{
 In this paper, new pairs of binary sequences with three cross correlation values are presented. The cross correlation values  are shown to be low. Finally we present some numerical
 results and some open problems. }

\medskip
\quad\emph{Index terms--} \textbf{Binary sequence,  Correlation
distribution, Linearized polynomial, Rank}



%

\section{Introduction}

Sequences with low cross correlations have wide applications in many different communication systems.
Let $u=(u_t)_{t=0}^{n-1}$ and $v=(v_t)_{t=0}^{n-1}$ be two binary sequences with length $n$.
The cross correlation function of $a$ and $b$ with shift $\tau$ is defined by
\begin{equation}\label{def cross correlation}
C_{\tau}(u,v)=\sum\limits_{t=0}^{n-1}(-1)^{u_{t+\tau}+v_t}.
\end{equation}
The multiset $\left\{C_{\tau}(u,v)\,|\, 0\leq \tau\leq n-1\right\}$ is called cross correlation distribution of the sequences $u$ and $v$. For several decades,
sequences with low cross correlations have attracted special interests from different aspects. In many cases the sequence $v$ is a decimated sequence of $u$ with some
decimation $d$, that is, $v_t=u_{(dt)}$ for $0\leq t\leq n-1$ where the subscript $(dt)$ takes on the smallest non-negative integer equalling $dt$ module $n$.
If $d$ and $n$
are coprime, then $v$ will have the same length as $u$ and both these two sequences are called $m$-sequences of length $n$.
For cross correlation of two binary $m$-sequences, together with a survey on historical results, the readers are referred to

In general the decimation $d$ may be not comprime to $n$. Then the length of $v$ will become $n'=n/\gcd(n,d)$ which is a factor of $n$.
In this case we can consider the cross correlation
function of $a$ and $b$ in a similar way as (\ref{def cross correlation}). In , Ness and Helleseth studied the cross correlation between a
$m$-sequence $(u_t)$ of length $2^m-1$ and its
decimated sequence $(v_t)=(v_{(dt)})$ of length $2^{m/2}-1$ with $d=(2^{m/2}+1)^2/3$\, (here $m\equiv 2\;(\mathrm{mod}\, 4)$).
Later, the result is generalized to the case
$d=(2^{(m+2)/4}-1)(2^{m/2}+1)$ (see ).  In , they propose a conjecture that all three-valued cross correlation between two $m$-sequences of lengths $2^m-1$ and $2^{m/2}-1$ are characterized.

In this paper we will consider the cross correlation between two $m$-sequences of length $2^m-1$ and $3(2^{m/2}-1)$ with $m\equiv 2\;(\mathrm{mod}\, 4)$.
The cross correlation is proven to be three-valued. Our result does not violate this conjecture since the second sequence has length $3(2^{m/2}-1)$, not $2^{m/2}-1$.
Our result reveals  that if the length of the second sequence $b$ is more flexible, there maybe exist more sequences with three-valued cross correlation.

Precisely, let $k$, $l$ be odd integers coprime to each other. Let $m=2k$, $q=2^m$ and $g$ be a primitive element of $\GF(q)$.
 We will study cross correlation of sequences
 \begin{equation}\label{def u}
 u=(\Tr_m(g^{t}))_{t=0}^{q-2}
 \end{equation}
  and its decimated sequence
  \begin{equation}\label{def v}
 v=(\Tr_m(g^{dt}))_{t=0}^{q-2}
 \end{equation}
 with decimation $d=(2^{lk}+1)/(2^l+1)$. Then by (\ref{def cross correlation}), the
 cross correlation function with shift $\tau$ is
 \begin{equation}\label{def dem cor}
 C_{\tau}(u,v)=\sum\limits_{t=0}^{n-1}(-1)^{\Tr_m(g^{t+\tau}+g^{dt})}
 =\sum\limits_{x\in \GF(q)^*}(-1)^{\Tr_m(x^d+ax)}
 \end{equation}
 with $a=g^{\tau}$. Our main result is depicted in the sequel.

 \begin{theo}\label{main result}
The cross correlation distribution $C_d(\tau)$ when $\tau$ runs from
$0$ to $q-2$ is as follows.

\[   \begin{array}{ccccc}
     \mathrm{values} & & &\mathrm{frequencies}&\\[2mm]
      -1 &&& \frac{(2^k+1)(7\cdot 2^{k}+8)}{9}&\\[2mm]
      -1+2^{k+1} &&& \frac{(2^k+1)^2}{9} &\\[2mm]
      -1-2^{k+1} &&& \frac{(2^k+1)(2^{k}-2)}{9}.&\qquad \square
    \end{array}
\]
\end{theo}

This paper is organized as follows. In Section II we will present some preliminaries which is needed to study the sequences, and also we will develop connections among the
cross correlation function and two kinds of exponential sums $T(a,b)$ and $S(a,b)$. In Section III we will study
$T(a,b)$ which is a kind of exponential sums from binary quadratic forms. In Section IV we will prove our main result. Finally in Section V we make some conclusions and also
some open problems will be proposed.

\section{Preliminaries}

The following notations are fixed through the rest of the paper
except for specific statements.
\begin{itemize}
  \item Let $k$ and $l$ be two odd integers with $0<l<k$ and
  $\gcd(l,k)=1$. Let $d=\frac{2^{lk}+1}{2^l+1}$.
  \item Let $m=2k$, $q=2^m$. For positive integer $i$,  let
  $\GF(2^i)$ be the finite field with cardinality $2^i$.
  \item Let $g$ be a fixed primitive element of $\GF(q)$ and
  $r=g^{2^k-1}$.
  \item  For $j|i$, let $\Tr_{i/j}:\GF(2^i)\ra \GF(2^j)$ be the trace
  mapping defined by
  $\Tr_{i/j}(x)=x+x^{2^j}+x^{2^{2j}}+\cdots+x^{2^{i-j}}$. In particular, we
  use the notation $\Tr_{i}$ to replace $\Tr_{i/1}$ for
  abbreviation.
\end{itemize}

If we regard $\GF(q)$ as a vector space over $\GF(2)$ with dimension $m$, then $Q(x)=\Tr_m\left(\sum\limits_{i}a_ix^{p^i+1}\right)$
is a binary quadratic form of $m$ variables.   It is well-known that $Q(x)$ is equivalent to one of the following three standard forms(see \cite{Lidl Niederreiter},
 Theorem ):
\begin{itemize}
  \item[](Type I): $x_1x_2+x_3x_4+\cdots+x_{2v-1}x_{2v}$;
  \item[] (Type II): $x_1x_2+x_3x_4+\cdots+x_{2v-1}x_{2v}+x_{2v+1}^2$;
  \item[](Type III): $x_1x_2+x_3x_4+\cdots+x_{2v-1}x_{2v}+x_{2v-1}^2+x_{2v}^2$
\end{itemize}
where $2v$ is codimension of $\GF(2)$-vector space $V_m$ which is defined by
\[
V_m=\left\{x\in \GF(q)\,|\, Q(x+y)+Q(x)+Q(y)=0\;\text{for all}\; y\in \GF(q)\right\}.
\]
Then $\sum\limits_{x_\in \GF(2)^m}(-1)^{Q(x)}$ can be evaluated if we know which standard form $Q(x)$ lies in.
\begin{lem}\label{three forms}
\[
\sum\limits_{x_\in \GF(2)^m}(-1)^{Q(x)}=\left\{
\begin{array}{ll}
2^{m-v}, &\text{if}\; Q(x)\;\text{belongs to Type I};\\[2mm]
0, &\text{if}\; Q(x)\;\text{belongs to Type II};\\[2mm]
-2^{m-v}, &\text{if}\; Q(x)\;\text{belongs to Type III}.
\end{array}
\right.
\]
\end{lem}

\begin{proof}
If $Q(x)$ belongs to Type I, then
\[\sum\limits_{x_\in \GF(2)^m}(-1)^{Q(x)}=2^{m-2v}\prod\limits_{i=1}^v\sum\limits_{x_{2i-1},x_{2i}=0}^1(-1)^{x_{2i-1}x_{2i}}=2^{m-2v}\cdot 2^v=2^{m-v}.\]

If $Q(x)$ belongs to Type II, then
\[\sum\limits_{x_\in \GF(2)^m}(-1)^{Q(x)}=2^{m-2v-1}\prod\limits_{i=1}^v\sum\limits_{x_{2i-1},x_{2i}=0}^1(-1)^{x_{2i-1}x_{2i}}\sum\limits_{x_{2v+1}=0}^1(-1)^{x_{2v+1}^2}=0\]
since the last summation equals to zero.

If $Q(x)$ belongs to Type III, then
\[\sum\limits_{x_\in \GF(2)^m}(-1)^{Q(x)}=2^{m-2v}\prod\limits_{i=1}^{v-1}\sum\limits_{x_{2i-1},x_{2i}=0}^1(-1)^{x_{2i-1}x_{2i}}\sum\limits_{x_{2v-1},x_{2v}=0}^1
(-1)^{x_{2v-1}x_{2v}+x_{2v-1}^2+x_{2v}^2}=-2^{m-v}\]
since the last summation equals to $-1$.
\end{proof}

Recall that $r=g^{2^k-1}$ and define $\de=r^d$.

\begin{lem}
The element $\de$ is  primitive in $\GF(4)$, that is,
$\de^2=\de+1=\de^{-1}$.
\end{lem}
\begin{proof}
We can calculate
$\de^{2^l+1}=g^{(2^k-1)(2^{lk}+1)}=g^{(2^{2k}-1)\cdot
\frac{2^{lk}+1}{2^k+1}}=1.$ By $\gcd(2^l+1,2^{2k}-1)=3$ we can
deduce $\de^3=1$.

It remains to show that $\de\neq 1$. Otherwise, by
$\de=g^{(2^k-1)\frac{2^{lk}+1}{2^l+1}}=1$ we have \[2^{2k}-1\mid
(2^k-1)\frac{2^{lk}+1}{2^l+1}\] which implies
\[2^k+1\mid \frac{2^{lk}+1}{2^l+1}.\]
So we can deduce
\begin{equation}\label{cont}
(2^k+1)(2^l+1)\mid 2^{lk}+1.
\end{equation}

 Denote by $v_3(n)$ the highest power of $3$ dividing $n$, that is, $n=3^{v_3(n)}n'$ with $n'$ coprime to $3$.
 For $n$ coprime to $a$, define
 \[\mathrm{ord}_n(a)=\min\{s>0|a^s\equiv1\;(\mathrm{mod}\, n)\}.\]
 Then $\mathrm{ord}_9(2)=6$.

 Then for odd $f$, we have
$v_3(2^f+1)\geq 1$ and then
\[
\begin{array}{rcl}
v_3(2^{3f}+1)&=&v_3((2^{f}+1)(2^{2f}-2^f+1))\\[2mm]
 &=&v_3(2^{f}+1)+v_3((2^{f}+1)^2-3\cdot 2^f)=v_3(2^{f}+1)+1
\end{array}
\]
where the last equality from $v_3((2^{f}+1)^2)\geq 2$ and $v_3(3\cdot 2^f)=1$. Then by induction we obtain
\begin{equation}\label{v3 2f}
v_3(2^f+1)=v_3(2^{3^{v_3(f)}f'}+1)=v_3(f)+v_3(2^{f'}+1).
\end{equation}
Since odd $f'$ is not divisible by $3$, then we can deduce $2^{f'}+1\not\equiv 0\;(\mathrm{mod}\, 9)$.  Otherwise
 $2^{2f'}\equiv 1\;(\mathrm{mod}\, 9)$
 which implies $6\mid 2f'$ and $3\mid f'$.
 It is a contradiction.
  Therefore $v_3(2^{f'}+1)=1$
 and then by (\ref{v3 2f}) we
 have $v_3(2^f+1)=v_3(f)+1.$
Then by (\ref{cont}) we obtain $v_3((2^k+1)(2^l+1))=1+v_3(k)+1+v_3(l)\leq v_3(2^{lk}+1)=1+v_3(lk)=1+v_3(l)+v_3(k)$
which leads to a contradiction.

As a result, we can deduce $\de\in \GF(4)\backslash\{0,1\}$ and then the result follows.

\end{proof}

By (\ref{def dem cor}) we obtain that
\begin{equation}
C_{\tau}(u,v)=S(a)-1
\end{equation}
with $a=g^{\tau}\in \GF(q)^*$ and
\begin{equation}\label{defn S}
S(a)=\sum\limits_{x\in \GF(q)}(-1)^{\Tr_m\left(x^d+ax\right)}.
\end{equation}

Then we turn to study the binary exponential sum $S(a)$. Firstly we observe that $\gcd(2^l+1,2^m-1)=3$ and
$(2^l+1)d=2^{lk}+1\equiv 2^k+1\pmod{2^m-1}$.  Since $r=g^{2^k-1}$ is
a noncube in $\GF(q)$, then images of the mappings $x\mapsto
x^{2^l+1}$, $x\mapsto rx^{2^l+1}$ and $x\mapsto r^{-1} x^{2^l+1}$
(all the three maps are from $\GF(q)$ to $\GF(q)$) covers each
element of $\GF(q)$ exactly three times when $x$ runs through $\GF(q)$. Therefore
\begin{eqnarray}\label{three one}
&3S(a)=\sum\limits_{x\in
\GF(q)}(-1)^{\Tr_m\left(x^{2^k+1}+ax^{2^l+1}\right)}+\sum\limits_{x\in
\GF(q)}(-1)^{\Tr_m\left(\de x^{2^k+1}+rax^{2^l+1}\right)}\nonumber\\[2mm]
&\qquad\qquad\qquad+\sum\limits_{x\in \GF(q)}(-1)^{\Tr_m\left(\de^{-1}
x^{2^k+1}+r^{-1}ax^{2^l+1}\right)}
\end{eqnarray}
where $\de=r^{d}.$

We observe that $x^{2^k+1}\in \GF(2^k)$ and $\de+\de^{2^k}=1$. Hence $\Tr_m(x^{2^k+1})=0$ and $\Tr_m(\de x^{2^k+1})=\Tr_k((\de+\de^{2^k})x^{2^k+1})=\Tr_k(x^{2^k+1})$.
In the same way $\Tr_m(\de^{-1} x^{2^k+1})=\Tr_k(x^{2^k+1})$.
It follows from (\ref{three one}) that
\begin{equation}\label{first form S}
3S(a)=T(a,0)+T(ra, \de)++T(r^{-1}a, \de)
\end{equation}
where
\begin{equation}\label{defn T}
T(a,b)=\sum\limits_{x\in
\GF(q)}(-1)^{\Tr_m\left(ax^{2^l+1}+bx^{2^k+1}\right)}.
\end{equation}

In order to evaluate $T(a,b)$, it is sufficient to study quadratic form
\begin{equation}\label{defn quadratic form}
Q_{a,b}(x):=\Tr_m\left(ax^{2^l+1}+bx^{2^k+1}\right).
\end{equation}

Define
\begin{equation}\label{defn Labx}
L_{a,b}(x)=a^{2^{2k-2l}}x^{2^{2k-2l}}+(b^{2^{2k-l}}+b^{2^{k-l}})x^{2^{k-l}}+a^{2k-l}x
\end{equation}
and
\begin{equation}\label{defn Vm(a,b)}
V_m(a,b)=\left\{x\in \GF(q)\,|\, Q_{a,b}(x+y)+Q_{a,b}(x)+Q_{a,b}(y)=0\;\text{for all}\; y\in \GF(q)\right\}.
\end{equation}

Now we turn to study $V_m(a,b)$ which can be formulated as follows.
\[
\begin{array}{ll}
&\dim_{\GF(2)}V_m(a,b)=m-2v \Longleftrightarrow \text{for all}\; y\in \GF(q),Q_{a,b}(x+y)+Q_{a,b}(x)+Q_{a,b}(y)=0\\[2mm]
& \qquad\qquad \qquad\qquad\qquad\qquad\qquad\text{has}\;2^{m-2v}\;\text{common solutions}\;x\in \GF(q)\\[2mm]
&\qquad\qquad \Longleftrightarrow \text{for all}\; y\in \GF(q),\Tr_m(ax^{2^l}y+axy^{2^l}+bx^{2^k}y+bxy^{2^k})=0\;\text{has}\;2^{m-2r}\\[2mm]
&\qquad\qquad \qquad\text{common solutions}\;x\in \GF(q)\\[2mm]
&\qquad\qquad \Longleftrightarrow \text{for all}\; y\in \GF(q),\Tr_m\left(y^{2^{2k-l}}(a^{2k-l}x+a^{2^{2k-2l}}x^{2^{2k-2l}}+(b^{2^{2k-l}}+b^{2^{k-l}})x^{2^{k-l}})\right)=0\\[2mm]
&\qquad\qquad \qquad\text{has}\;2^{m-2v}\;\text{common solutions}\;x\in \GF(q)\\[2mm]
&\qquad\qquad \Longleftrightarrow L_{a,b}(x)=a^{2^{2k-2l}}x^{2^{2k-2l}}+(b^{2^{2k-l}}+b^{2^{k-l}})x^{2^{k-l}}+a^{2k-l}x=0\\[2mm]
&\qquad\qquad\qquad\text{has}\;2^{m-2v}\;\text{solutions}\;x\in \GF(q).
\end{array}
\]

Therefore
\[V_m(a,b)=\left\{x\in \GF(q)\,|\,L_{a,b}(x)=0\right\}.\]
Note that $\gcd(k-l,m)=2$ and $V_m(a,b)$ is also a $\GF(2^{m})\cap \GF(2^e)=\GF(4)$-linear space. Then we can determine the possible dimensions of $V_m(a,b)$.

\begin{lem}
For any $a\in \GF(q)^*$, the dimension of  $\GF(4)$-linear space
$V_{m}(a,b)$ is at most 2.
\end{lem}
\begin{proof}
Fix an algebraic closure $\GF({2^\infty})$ of $\GF(2)$, since
the degree of $\GF(4)$-linearized polynomial $L_{a,b}(x)$ is
$2^{2e}$ and $L_{a,b}(x)=0$ has no multiple roots in
$\GF({2^\infty})$,  then the zeroes of $L_{a, b}(x)$ in
$\GF({p^\infty})$, say $V_{\infty}(a,b)$, form an $\GF(2^e)$-vector
space of dimension 2. Note that $\gcd(e,m)=2$. Then
$V_m(a,b)=V_{\infty}(a,b)\cap \GF(2^m)$ is a vector space on
$\GF(2^{\gcd(e,m)})=\GF(4)$ with dimension at most 2 since any
elements in $\GF(2^m)$ which are linear independent over $\GF(4)$
are also linear independent over $\GF(2^e)$.
\end{proof}

The following binary exponential sum will be useful(see \cite{Lidl Niederreiter}, \cite{Moisio}).
\begin{lem}\label{primitive}
For $h\mid 2^k+1$, we have
\[
\sum\limits_{x\in \GF(q)}(-1)^{\Tr_m(ax^h)}=\left\{
\begin{array}{ll}
(h-1)2^k,&\text{if}\; a=g^{hi}\;\text{for some}\;i,\\[2mm]
-2^k,&\text{otherwise}.\qquad\square
\end{array}
\right.
\]
\end{lem}

Now we introduce some moment identities to determine the occurrences
of all possible values of $S(a)$.

\begin{lem}\label{moment}
For $S(a)$ defined in (\ref{defn S}), we have

 (i). $\sum\limits_{a\in \GF(q)^*}S(a)=\frac{2^k+1}{3}\cdot 2^{k+1}$.

 (ii). $\sum\limits_{a\in \GF(q)^*}S(a)^2=\frac{2^{2k+2}(2^k+1)(2^{k+1}-1)}{9}$.
\end{lem}
\begin{proof}
(i). We calculate
\[\begin{array}{ll}
&\sum\limits_{a\in \GF(q)^*}S(a)=\sum\limits_{x\in
\GF(q)}(-1)^{\Tr_m\left(x^d\right)}\sum\limits_{a\in
\GF(q)^*}(-1)^{\Tr_m\left(ax\right)}\\[2mm]
&\qquad\qquad\quad =2^{2k}-1-\sum\limits_{x\in
\GF(q)^*}(-1)^{\Tr_m\left(x^d\right)}\\[2mm]
&\qquad\qquad\quad=2^{2k}-1-\left(\left(\frac{2^k+1}{3}-1\right)2^k-1\right)=\frac{2^k+1}{3}\cdot
2^{k+1}.
\end{array}\]

(ii). We obtain
\[\begin{array}{ll} &\sum\limits_{a\in
\GF(q)^*}S(a)^2=\sum\limits_{x,y\in
\GF(q)}(-1)^{\Tr_m\left(x^d+y^d\right)}\sum\limits_{a\in
\GF(q)^*}(-1)^{\Tr_m\left(a(x+y)\right)}\\[2mm]
&\qquad\qquad\quad
=\left(2^{2k}-1\right)\sum\limits_{x=y}(-1)^{\Tr_m\left(x^d+y^d\right)}-\sum\limits_{x\ne
y}(-1)^{\Tr_m\left(x^d+y^d\right)}.
\end{array}\]
Denote by $A=\sum\limits_{x=y}(-1)^{\Tr_m\left(x^d+y^d\right)}$ and
$B=\sum\limits_{x\ne y}(-1)^{\Tr_m\left(x^d+y^d\right)}$. Then
$A=2^{2k}$ and
\[A+B=\sum\limits_{x,y}(-1)^{\Tr_m\left(x^d+y^d\right)}=\left(\sum\limits_{x\in \GF(q)}(-1)^{\Tr_m\left(x^d\right)}\right)^2=\frac{2^{2k}(2^k-2)^2}{9}.\]
Therefore $B=2^{2k}(2^k-5)(2^k+1)/9$ and
\[\sum\limits_{a\in
\GF(q)^*}S(a)^2=2^{2k}\left(2^{2k}-1\right)-\frac{2^{2k}(2^k-5)(2^k+1)}{9}=\frac{2^{2k+2}(2^k+1)(2^{k+1}-1)}{9}.\]
\end{proof}

\section{On binary exponential sum $T(a,b)$}

For any $x\in \GF(q)^*$ and  $0\leq i\leq 2$, it is easy to see
$(\de^ix)^{2^l+1}=x^{2^l+1}$ and $(\de^ix)^{2^k+1}=x^{2^k+1}$. Hence
\[\Tr_m\left(a(\de^ix)^{2^l+1}+b(\de^ix)^{2^k+1}\right)=\Tr_m\left(ax^{2^l+1}+bx^{2^k+1}\right).\]
As a consequence,
\begin{lem}\label{T(a,b) equiv 1 mod 3}
For any $a,b\in \GF(q)$, we have $T(a,b)\equiv 1\pmod{3}.$
\end{lem}
\begin{proof}
Let $D$ be a set of coset representatives of
$\GF(q)^*\big{/}\GF(4)^*$. Then
\[
\begin{array}{ll}
&T(a,b)=1+\sum\limits_{x\in D}\sum\limits_{i=0}^{2}(-1)^{\Tr_m(a(\de^i x)^{2^l+1}+b(\de^i x)^{2^k+1})}\\[2mm]
&\qquad\quad\;=1+3\cdot\sum\limits_{x\in
D}(-1)^{\Tr_m(ax^{2^l+1}+bx^{2^k+1})}\equiv 1\,(\mathrm{mod}\; 3).
\end{array}
\]
\end{proof}

Now we can decide the possible values of $T(a,b)$.

\begin{lem}\label{exp pol}
The exponential sum
\[
T(a,b)=\left\{
\begin{array}{ll}
-2^k, &\text{if}\;  \dim_{\GF(4)}V_m(a,b)=0,\\[2mm]
2^{k+1}, &\text{if}\;  \dim_{\GF(4)}V_m(a,b)=1,\\[2mm]
-2^{k+2}, &\text{if}\;  \dim_{\GF(4)}V_m(a,b)=2.
\end{array}
\right.
\]
\end{lem}
\begin{proof} If $\dim_{\GF(4)}V_m(a,b)=0$, then $m-2v=0$ and $v=k$. Hence by Lemma \ref{three forms} we obtain $T(a,b)=\pm 2^k$. Combining  Lemma \ref{T(a,b) equiv 1 mod 3} we  can deduce $T(a,b)=-2^k$.

If $\dim_{\GF(4)}V_m(a,b)=1$, then $\dim_{\GF(2)}V_m(a,b)=2$ and $m-2v=2$ which yields $v=k-1$. Hence by Lemma \ref{three forms} we obtain $T(a,b)=\pm 2^{k+1}$ or $0$. Combining  Lemma \ref{T(a,b) equiv 1 mod 3} we obtain $T(a,b)=2^{k+1}$.

Similarly, if $\dim_{\GF(4)}V_m(a,b)=2$, then $T(a,b)=-2^{k+2}$.
\end{proof}

To evaluate $S(a)$, we need to deal with $T(a,0), T(ra,\de)$ and
$T(r^{-1}a,\de^{-1})$ simultaneously.
\begin{lem}\label{T(a,0)}
For $a\in \GF(q)^*$, we have
\[
T(a,0)=\left\{
\begin{array}{ll}
2^{k+1}, &\text{if}\; a\;\text{is a cubic},\\[2mm]
-2^k, &\text{if}\; a\;\text{is not a cubic}.
\end{array}
\right.
\]
\end{lem}
\begin{proof}
In this case, $L_{a,0}(x)=a^{2^{2e}}x^{2^{2e}}+a^{2^{k+e}}x=0$ has
nonzero solution in $\GF(q)$ if and only if
\[x^{2^{2e}-1}=\left(a^{2^l-1}\right)^{2^{2e}}.\]
Since $\gcd(2^{2e}-1,q-1)=3$ and $\gcd(2^l-1,q-1)=1$, this equation
has nonzero solution if and only if $a$ is a cubic in $\GF(q)^*$. In
this case, it has exactly three nonzero solutions. Taking the
solution $x=0$ into account, we obtain that $L_{a,0}(x)=0$ has four or
one solutions in $\GF(q)$ depending on $a$ in cubic in $\GF(q)^*$ or
not. Therefore the result follows from Lemma \ref{exp pol}.
\end{proof}

\begin{lem}\label{T(ra,del) in cubic}
If $a$ is a nonzero cubic in $\GF(q)^*$, then
\[T(ra,\de)=T(r^{-1}a,\de^{-1})=-2^k\;\text{or}\; 2^{k+1}.\]
\end{lem}
\begin{proof}
Firstly we show that $T(ra,\de)=T(r^{-1}a,\de^{-1})$. Note that
$\de^{2^k}=\de^2=\de^{-1}$ and $\de+\de^2=1$. We can reformulate
\[T(ra,\de)=\sum\limits_{x\in
\GF(q)}(-1)^{\Tr_m\left(rax^{2^l+1}\right)+\Tr_k\left(x^{2^k+1}\right)}\]
and
\[T(r^{-1}a,\de)=\sum\limits_{x\in
\GF(q)}(-1)^{\Tr_m\left(r^{-1}ax^{2^l+1}\right)+\Tr_k\left(x^{2^k+1}\right)}.\]
Since $a$ is a nonzero cubic in
$\GF(q)$, we assume $a=g^{3s}$ for some integer $s$. It is easy to
see that $\gcd((2^k-1)(2^l+1), 2^k+1)=3$. Hence there exists
integers $i$ and $j$ satisfying
\[(2^k-1)(2^l+1)i+(2^k+1)j=3s.\]
By substituting $x=g^{-(2^k-1)i}y$, we obtain
$ax^{2^l+1}=g^{3s-(2^k-1)(2^l+1)i}y^{2^l+1}=g^{(2^k+1)j}y^{2^l+1}$ and $x^{2^k+1}=y^{2^k+1}$.
Denote by $b=g^{(2^k+1)j}\in \GF(2^k)^*$. Then $T(ra,\de)=T(rb,\de)$
 and $T(r^{-1}a,\de)=T(r^{-1}b,\de)$. Moreover
 \[\begin{array}{ll}
 &T(rb,\de)=\sum\limits_{x\in
\GF(q)}(-1)^{\Tr_m\left((rbx^{2^l+1})^{2^k}\right)+\Tr_k\left(x^{2^k+1}\right)}\\[2mm]
&\qquad =\sum\limits_{x\in
\GF(q)}(-1)^{\Tr_m\left(r^{-1}bx^{2^k(2^l+1)}\right)+\Tr_k\left(x^{2^k+1}\right)}=T(r^{-1}b,\de).
\end{array}
\]
Therefore $T(ra,\de)=T(rb,\de)=T(r^{-1}b,\de)=T(r^{-1}a,\de)$. From now on we may assume $a\in \GF(2^k)$.

Secondly we show that $T(ra,\de)\ne -2^{k+2}$ which is equivalent to
saying that $\dim_{\GF(4)}V_{ra,\de}\neq 2$.
Assume, on the contrary, that $\dim_{\GF(4)}V_{ra,\de}=2$. Then
there exists $x_1, x_2$ with $x_1\neq x_2, \de x_2, \de^2 x_2$.
Thereafter
\[(ra)^{2^{2e}}x_1^{2^{2e}}+x_1^{2^e}+(ra)^{2^{k+e}}x_1=(ra)^{2^{2e}}x_2^{2^{2e}}+x_2^{2^e}+(ra)^{2^{k+e}}x_2=0\]
which yields
\[\left((ra)^{2^{2e}}x_1^{2^{2e}}+(ra)^{2^{k+e}}x_1\right) x_2^{2^e}=\left((ra)^{2^{2e}}x_2^{2^{2e}}+(ra)^{2^{k+e}}x_2\right) x_1^{2^e}.\]
A routine calculation implies that
\[a^{2^e(2^e-1)}\left(x_1^{2^e}x_2+x_1x_2^{2^e}\right)^{1-2^e}=r^{2^e(2^e+1)}.\]
The left hand side is a cubic in $\GF(q)$. But the right hand side
is not since $r=g^{2^k-1}$ and $3\nmid \gcd((2^k-1)(2^e+1),q-1)=2^k-1$. It leads to a
contradiction.
\end{proof}

In the sequel we will consider $T(ra,\de)$ and $T(r^{-1}a,\de)$ in
the case $a$ is noncubic.
\begin{lem}\label{T noncubic}
If $a$ is a noncubic, then at least one of $T(ra,\de)$ and
$T(r^{-1}a,\de)$ is equal to $-2^k$.
\end{lem}
\begin{proof}
It suffices to show that at least one of $L_{ra,\de}(x)=0$ and
$L_{r^{-1}a,\de}(x)=0$ has only one solution $x=0$ in $\GF(q)$. Indeed, assume
there exist $x_1,x_2\in \GF(q)^*$ such that
\[(ra)^{2^{2e}}x_1^{2^{2e}}+x_1^{2^e}+(ra)^{2^{k+e}}x_1=(r^{-1}a)^{2^{2e}}x_2^{2^{2e}}+x_2^{2^e}+(r^{-1}a)^{2^{k+e}}x_2=0\]
which implies that
\[\left((ra)^{2^{2e}}x_1^{2^{2e}}+(ra)^{2^{k+e}}x_1\right)x_2^{2^e}=\left((r^{-1}a)^{2^{2e}}x_2^{2^{2e}}+(r^{-1}a)^{2^{k+e}}x_2\right)x_1^{2^e}.\]
It can be transformed to
\[\left(r^{-2^e}x_1^{2^e}x_2+r^{2^e}x_1x_2^{2^e}\right)^{2^e-1}=a^{2^{2e}(2^l-1)}.\]
Note that $\gcd(2^e-1,q-1)=3$ and $\gcd(2^l-1,q-1)=1$. Then $a$ must be
a cubic which is a contradiction.
\end{proof}

\section{Proof of main result}

Now we are ready to give all the possible values of $S(a)$ for $a\in
\GF(q)^*$.

\begin{lem}\label{T(ra,de) in noncubic}
For $a\in \GF(q)^*$, the possible values of $S(a)$ are $0, 2^{k+1},
-2^{k+1}$ and $-2^k$. Precisely,
\begin{itemize}
  \item[Case I:] if $a$ is nonzero cubic, then $S(a)=0$ or $2^{k+1}$;
  \item[Case II:] if $a$ is noncubic, then $S(a)=0, -2^k$ or $-2^{k+1}$.
\end{itemize}
\end{lem}
\begin{proof}
If $a$ is a nonzero cubic, then $T(a)=2^{k+1}$. From Lemma \ref{T(ra,del) in cubic} we have $T(ra,\de)=T(r^{-1}a,\de)=-2^k$ or $2^{k+1}$. As a consequence,
$S(a)=\left(2^{k+1}-2^k-2^k\right)/3=0$ or $S(a)=\left(2^{k+1}+2^{k+1}+2^{k+1}\right)/3=2^{k+1}$.

If $a$ is a noncubic, then $T(a)=-2^k$. By Lemma \ref{T noncubic} we
obtain $S(a)=\left(-2^k-2^k-2^k\right)/3=-2^k$ or
$S(a)=\left(-2^k-2^{k}+2^{k+1}\right)/3=0$ or
$S(a)=$ \\$\left(-2^k-2^{k}-2^{k+2}\right)/3=-2^{k+1}$.
\end{proof}

When $a$ runs through $\GF(q)^*$, suppose $S(a)$ takes on the value
zero $N_0$ times, $-2^k$ is taken on $N_1$ times, $2^{k+1}$ is taken
on $N_2$ times and $-2^{k+1}$ is taken on $N_3$ times. Since
$2^{k+1}$ only occurs in Case I (see Lemma \ref{T(ra,de) in
noncubic}), then $N_2$ can be calculated directly.
\begin{lem}
\[N_2=\frac{(2^k+1)^2}{9}.\]
\end{lem}
\begin{proof}
\begin{eqnarray}\label{eqn N2}
&&2^{k+1} N_2=\sum\limits_{a\,\text{nonzero cubic}}\sum\limits_{x\in
\GF(q)}(-1)^{\Tr_m\left(x^d+ax\right)}\nonumber\\[2mm]
&&\qquad =\frac{1}{3}\sum\limits_{x\in
\GF(q)}(-1)^{\Tr_m\left(x^d\right)}\sum\limits_{b\in
\GF(q)^*}(-1)^{\Tr_m\left(b^3x\right)}\nonumber\\[2mm]
&& \qquad
=\frac{1}{3}\left(q-1+(2^{k+1}-1)\sum\limits_{x\,\text{nonzero
cubic}}(-1)^{\Tr_m\left(x^d\right)}+(-2^{k}-1)\sum\limits_{x\,\text{non
cubic}}(-1)^{\Tr_m\left(x^d\right)}\right)\nonumber\\[2mm]
&& \qquad =\frac{1}{3}\left(q-1+(2^{k+1}-1)\cdot A+(-2^{k}-1)\cdot
B\right)
\end{eqnarray}

where $A=\sum\limits_{x\,\text{nonzero
cubic}}(-1)^{\Tr_m\left(x^d\right)}$ and
$B=\sum\limits_{x\,\text{non cubic}}(-1)^{\Tr_m\left(x^d\right)}$.

Since when $x$ runs through $\GF(q)^*$, $x^{2^l+1}$ runs through each nonzero cubic in $\GF(q)$ exactly three times, then we can calculate
\[A=\frac{1}{3}\sum\limits_{y\in \GF(q)^*}(-1)^{\Tr_m\left(y^{2^k+1}\right)}=\frac{q-1}{3}\]
and
\[A+B=\sum\limits_{x\in \GF(q)^*}(-1)^{\Tr_m\left(x^d\right)}=\sum\limits_{x\in \GF(q)^*}(-1)^{\Tr_m\left(x^{(2^k+1)/3}\right)}=2^k(2^k-2)/3-1.\]
Substituting $A$ and $B$ into (\ref{eqn N2}), we obtain
\begin{equation}\label{N2}
N_2=\frac{(2^k+1)^2}{9}.
\end{equation}
\end{proof}

Now we are ready to determine the cross correlation of the sequences $u$ defined in (\ref{def u}) and
$v$ defined in (\ref{def v}).

{\it Proof of Theorem 1:}
Recall the definitions of $N_i$ ($0\leq i\leq 2$) and the value of
$N_2$ in (\ref{N2}). Then we have

\begin{equation}\label{moment 0}
N_0+N_1+N_3=2^{2k}-1-N_2=\frac{(2^k+1)(2^{k+3}-10)}{9}.
\end{equation}
From Lemma \ref{moment} we obtain
\begin{equation}\label{moment 1}
N_1+2N_3=\frac{2(2^k+1)(2^{k}-2)}{9}
\end{equation}

\begin{equation}\label{moment 2}
N_1+4N_3=\frac{4(2^k+1)(2^{k}-2)}{9}.
\end{equation}

Solving the system equations consisting of (\ref{moment
0})-(\ref{moment 2}), we can calculate
\[N_0=\frac{(2^k+1)(7\cdot 2^{k}+8)}{9},\qquad N_1=0, \qquad N_3=\frac{(2^k+1)(2^{k}-2)}{9}.\qquad\square\]
Combing Lemma (\ref{N2}) we complete the proof of Theorem \ref{main result}.

\section{Conclusion}

\quad\; In this paper, we studied the cross correlation between one $m$-sequence of length $2^{2k}-1$ and its decimated sequence of
length $3(2^k-1)$. The cross correlation has three possible values: $0$, $2^{k+1}$, $-2^{k+1}$. Moreover, the cross correlation distribution is
also determined.




%

\end{document}